\documentclass[12pt,a4paper]{article}

\makeatletter

\renewcommand\section{\@startsection
 {section}{1}{0pt}
  {-\baselineskip}
  {0.5\baselineskip}
{\raggedright\reset@font\normalsize\bfseries\mathversion{bold}}}

\renewcommand{\subsection}{\@startsection
{subsection}{2}{0mm}
  {-\baselineskip}
{0.5\baselineskip}
{\normalfont\normalsize\itshape}}

\renewcommand{\subsubsection}{\@startsection
{subsubsection}{3}{0mm}
  {-\baselineskip}
  {0.3\baselineskip}
{\normalfont\normalsize\itshape}}

\long\def\@makefntext#1{%
  \noindent \hb@xt@ 0.2em{\hss \@makefnmark}\hskip0.5em\relax#1}%

\makeatother

\usepackage{lmodern}
\usepackage{amsmath,amssymb,amsthm}
\usepackage[USenglish,british,american,australian]{babel}
\usepackage{wrapfig}
\usepackage{enumerate}
\usepackage{graphicx,color}
\usepackage{natbib}
\bibpunct[,]{(}{)}{,}{a}{}{,}
\usepackage{setspace}
\usepackage[margin=1in]{geometry}
\usepackage{soul}
\usepackage[title]{appendix}
\usepackage{tikz}
\usepackage{float}

\newcommand{\Avg}{\mathrm{Avg}}

\let\emptyset\varnothing

\theoremstyle{plain}
\newtheorem{theorem}{Theorem}
\newtheorem{corollary}{Corollary}
\newtheorem{lemma}{Lemma}
\newtheorem{proposition}{Proposition}

\theoremstyle{definition}
\newtheorem{definition}{Definition}
\newtheorem{example}{Example}
\newtheorem{remark}{Remark}
\newtheorem{assumption}{Assumption}

\newcounter{num}

\begin{document}
\bibliographystyle{plainnat}

\title{Full Subgame-Perfect Implementation of Optimal Risk Sharing on an Infinite Menu\footnote{An earlier version of this paper circulated under the title ``Risk Sharing Among Many: Implementing a Subgame-Perfect and Optimal Equilibrium.''}
}
\author{Michiko Ogaku\footnote{Faculty of Economics, Nagasaki University, 4-2-1 Katafuchi, Nagasaki, 850-8506 Japan}
}
\date{ }
\maketitle


\begin{abstract}
This paper extends the price-and-choose (P\&C) mechanism of
\citet{echenique2025price} to risk sharing in finance and insurance.
Cash invariance of monetary utilities makes the quasi-linear framework
applicable when P\&C prices are interpreted as deterministic monetary
transfers. We model each allocation as an essentially bounded random
vector and use a weak$^\ast$-compact infinite menu that allows allocations
to depend on background factors affecting aggregate loss.
Under complete information and prior commitment to the procedure,
if monetary utilities are weak$^\ast$-continuous on the menu, the extended
P\&C mechanism fully implements, in subgame-perfect equilibrium, the set
of allocations that maximize the sum of participants' monetary utilities.

This implementation result also applies when participants disagree
about probability assessments.
Each participant has their own weak$^\ast$-compact credal set of priors
dominated by a reference probability, allowing for finitely additive
charges, and their own countably additive probability measure over that
set. We define each participant's monetary utility as an exponential
certainty equivalent of the prior-specific utility evaluations, based on
integration with respect to this probability measure.
This specification extends the monetary certainty-equivalent
representation of the smooth-ambiguity model of \citet{Klibanoff2005}
to a broader class of utility evaluations,
while retaining exponential ambiguity evaluation.
If the prior-specific utilities are jointly continuous in
allocations and priors, this monetary utility is weak$^\ast$-continuous
on the choice set. Under complete information, with common knowledge
of each participant's probability assessment structure, the preceding
implementation result applies.
\end{abstract}


\newpage 
\section{Introduction}
We consider optimal risk sharing within a syndicate whose participants
jointly underwrite risks. The syndicate transfers part of its risks to
outside markets, while risks for which no outside counterparty can be
found must be shared internally. In particular, exposures to large
catastrophic losses often have to be shared among the participants.

In practice, risk sharing is supported by various institutions and
conventions, including delegation to professionals and sharing agreements
accepted in advance. At Lloyd's, for example, delegation to professionals
is combined with contractual participation shares \citep{lloyds2023agency}, while the International
Group of P\&I Clubs uses an annually renewed written Pooling Agreement to govern loss sharing
among clubs \citep{igpandiGroupAgreements}.

This paper studies a decentralized procedure, formalized as a mechanism,
as a possible alternative to these ways of determining allocations.
Participants who differ in their assessments of the likelihood and
magnitude of large losses can directly influence their own allocations
through their choices under a procedure to which they have agreed in
advance, rather than delegating allocation decisions to professionals or
contracting in advance on the allocation itself. Can they thereby agree
on a socially optimal allocation of the loss?

To address this question, we investigate the applicability of the
price-and-choose (P\&C) mechanism of \citet{echenique2025price}, hereafter
EN. Under complete information and prior commitment to the procedure, we
identify conditions under which an optimal risk allocation can be
implemented in subgame-perfect equilibrium.

\vspace{12pt}

In the P\&C mechanism, participants sequentially post price schedules over
a set of feasible allocations, and the final mover selects an allocation.
Once an allocation is selected, monetary transfers are executed according
to the posted price schedules.
EN show that, with a finite choice set and quasi-linear utilities, the
mechanism implements Pareto-optimal allocations in subgame-perfect
equilibrium. They also analyze extensions that relax quasi-linearity.
Our key observation is that this quasi-linear structure is naturally
compatible with monetary utilities in finance and insurance. Monetary
utilities are commonly defined to be cash invariant, so that a sure
monetary transfer enters utility additively. Thus, interpreting prices and
bids in P\&C as deterministic monetary transfers allows EN's implementation
mechanism to be applied to risk-sharing problems involving complex tail
risks. Risk-sharing problems, however, differ from EN's finite-menu
environment. The set of feasible allocations is generally not finite, and
a redistribution of the aggregate loss $X=\sum_i X_i$ among participants
is represented by a vector of random variables. We formulate each
allocation as an essentially bounded random vector in $(L^\infty)^n$ and
extend P\&C to an infinite menu.

\vspace{12pt}

\begingroup
\clubpenalty=10000
\widowpenalty=10000
We further equip the allocation menu with the weak$^\ast$ topology to
accommodate complex tail risks. This allows us to preserve a broad menu of
allocations that can depend on background factors affecting the aggregate
loss. Requiring norm compactness would generally impose
additional restrictions on admissible allocations. Specifically, the
choice set is a weak$^\ast$-compact set of feasible allocations whose
components have the same sign as the aggregate loss. On this set, we
consider weak$^\ast$-continuous price schedules normalized to have zero
mean. We also assume that each participant's monetary utility is
weak$^\ast$-continuous on the choice set. A concrete class of monetary
utilities satisfying this continuity condition can be constructed using
an exponential version of the smooth-ambiguity approach of
\citet{Klibanoff2005}.
Under these conditions, we show that the extended P\&C mechanism fully
implements, in subgame-perfect equilibrium, the set of allocations that
maximize the sum of participants' monetary utilities.

\vspace{12pt}

This implementation result also applies when participants disagree
about probability assessments. Each participant has their own
weak$^\ast$-compact credal set of priors dominated by a reference probability
$\mathbb P$. Candidate priors are allowed to be finitely additive
probability charges. Each participant also has their own countably
additive probability measure over their credal set. We define each
participant's monetary utility as an exponential certainty equivalent of
the prior-specific utility evaluations, based on integration with respect
to this probability measure.
This specification extends the monetary certainty-equivalent
representation of the smooth-ambiguity model of \citet{Klibanoff2005}
to a broader class of utility evaluations,
while retaining exponential ambiguity evaluation.
If the prior-specific utilities are jointly continuous in allocations
and priors, the resulting monetary utility is weak$^\ast$-continuous on the
choice set. Consequently, under complete information, with common knowledge of
the primitives underlying each participant's evaluations, the preceding
implementation result applies.

\vspace{12pt}

Finally, we define fairness in P\&C as equal division of the surplus above
participants' respective average utility levels. We apply EN's
Proposition~3 on the first-mover auction to our environment to address
the inequity inherent in P\&C. Under the basic P\&C mechanism, the
equilibrium payoffs of participants other than the first mover are fixed
at their average utility levels, while the first mover captures the
entire ex-ante surplus generated by an efficient allocation.
Since EN's bidding argument applies without modification under our
implementation results, we obtain, as a corollary based on EN's
Proposition~3, that P\&C with the first mover selected by bidding fully
implements optimal and fair risk sharing even when participants
disagree about probability assessments.

\endgroup
\subsection{Related Literature}

\subsubsection{Risk sharing under monetary utilities}

By showing that a risk allocation maximizing the sum of participants'
monetary utilities can be implemented in subgame-perfect equilibrium,
this paper contributes to the literature on convex risk measures and
monetary utility functions
\citep{artzner1999application,delbaen2002coherent,
follmer2016stochastic,kaina2009convex}.
Within this framework, \citet{jouini2008optimal} study optimal risk
sharing between two agents endowed with monetary utility functions on
$L^\infty$. Under law invariance, they establish the existence of an
optimal allocation in which each component is a nondecreasing function
of aggregate risk.

Relatedly, \citet{filipovic2008optimal} establish existence and
characterization results for optimal capital and risk allocations for
law-invariant, cash-invariant convex functionals on
$L^p$, $p\in[1,\infty]$.
They further show that an optimal risk allocation can be attained through
contracts whose components are nondecreasing Lipschitz-continuous
functions of aggregate risk.
\citet{dana2011comonotonicity} establishes the existence of efficient
risk allocations and competitive equilibria for strictly concave
utilities on $L^\infty$ that preserve second-order stochastic dominance.
In the setting of \citet{dana2011comonotonicity}, this class of utilities coincides with the class of
strictly concave law-invariant utilities.

Here, law invariance requires a participant to assign the same utility to
allocations with the same marginal distribution, even when their
dependence on background risks differs. Under preference assumptions
that include this condition, these studies reduce the problem, without
loss of optimality, to allocations that are nondecreasing functions of
aggregate risk. They then establish the existence of optimal allocations
by exploiting the norm compactness of suitably restricted families of
these allocations. However, under such allocations, each participant's
burden is identical across states with the same realization of aggregate
risk, even if background risks differ. Without imposing law invariance,
we retain an allocation menu that allows dependence on such background
states and analyze implementation under weak$^\ast$ compactness of the
menu and weak$^\ast$ continuity of monetary utilities.

\subsubsection{Subgame-perfect implementation and price-and-choose}

This paper is not the first to establish subgame-perfect implementation
in a related risk-sharing environment.
\citet{moore1988subgame} provide a classical benchmark for
subgame-perfect implementation under complete information.
They show that extensive-form stage mechanisms are very powerful for
implementing desirable social choices in economic environments with a
divisible private good, and they also discuss an application to pure
risk sharing.
Under condition $C^+$ for three or more agents and condition $C^{++}$
for two agents, they construct a general announcement-and-test mechanism
for implementing a desired social choice correspondence.

When the target social choice correspondence, that is, the mapping from
type profiles to sets of admissible allocations, is multivalued,
however, their general mechanism is not designed to prevent the
coordination failures that may arise when agents must decide which
admissible outcome to announce.
Once agents coordinate on a particular allocation, that allocation can
be implemented in equilibrium.

P\&C can likewise be interpreted as a mechanism for implementing a social
choice correspondence, but it does not require simultaneous announcements
of admissible allocations.
The realized allocation is selected by the final mover, whose choice is
disciplined by the sequentially posted price schedules.
P\&C therefore avoids the coordination problem specific to simultaneous
announcements.

\subsubsection{Foundational work on risk sharing and the structure of
efficient allocations}

Classical risk-sharing theory clarifies the structure of efficient
allocations but leaves open the problem of selecting one allocation from
among multiple efficient alternatives.
\citet[\ Section 2.5--2.6]{borch1962equilibrium} characterizes Pareto-optimal reinsurance
treaties in terms of insurers' utility functions and positive Pareto
weights.
\citet{borch1962equilibrium} shows that every Pareto-optimal treaty is equivalent to an arrangement
that pools the participants' individual risks and then allocates the
aggregate loss among them.
Because different choices of Pareto weights generate infinitely many
Pareto-optimal sharing rules, however, Pareto optimality alone does not
determine which rule should be adopted.
\citet{borch1962equilibrium} interprets this selection problem as an
$n$-person cooperative game involving negotiation among the participants.

\citet{wilson1968theory} analyzes Pareto-optimal sharing rules in a
syndicate whose members each hold a single, potentially different,
probability assessment.
His primary concern is not the existence of an efficient allocation
itself, but the conditions under which the collective decision making
induced by a Pareto-optimal sharing rule can be represented by a single
surrogate utility function and a single surrogate probability assessment.
This paper differs by studying a procedure that implements an optimal
allocation through the choices of participants who hold different
probability assessments.

Subsequent research has developed along several directions.
To address the failure of Pareto optimality alone to select a unique
risk-sharing rule, \citet{buhlmann1979optimal} propose a long-run fairness
condition as an additional selection criterion.
As a market-based approach to selecting an allocation from the Pareto set,
\citet{aase1993equilibrium} formulates a reinsurance syndicate as a
pure-exchange economy under uncertainty.
\citet{aase1993equilibrium} provides conditions for the existence and characterization of a
competitive equilibrium and discusses conditions for the uniqueness of the
equilibrium and the equilibrium premium.

\subsubsection{Mechanism design and implementation under ambiguity}

Much of the literature on mechanism design and implementation under
ambiguity analyzes how ambiguity affects behavior and implementability
under allocation and transfer rules specified by a designer. Examples
include auction environments with ambiguous information about the
distribution of valuations \citep{chen2007sealed,bose2009dynamic} and
non-Bayesian models of beliefs in mechanism design
\citep{chen2015mechanism}.
An exception is the decentralized P\&C mechanism of
\citet{echenique2025price}, in which participants themselves post prices
and choose the allocation. EN discuss the robustness of its implementation
result when an opponent's approximately optimal behavior is evaluated
according to a max--min criterion. We extend this decentralized procedure
to a risk-sharing environment with heterogeneous credal sets over the
uncertain events affecting aggregate risk and an infinite menu of allocations.


The remainder of this paper is organised as follows. Section \ref{sec:preliminaries} defines the primitives and what efficient means. 
Section \ref{sec:P&C} shows how the P\&C mechanism yields Pareto-optimal risk sharing in both two-player and multi-player settings.
Section \ref{sec:credal-set} extends the P\&C risk-sharing model to
exponential smooth ambiguity with heterogeneous credal sets and
beliefs over priors.
Section \ref{sec:bidding} describes how the first-mover auction improves the fairness of P\&C risk sharing under the multiple-priors environment. Section \ref{sec:conclusion} concludes. 

\section{Preliminaries: Environment and efficient benchmark}\label{sec:preliminaries}
\subsection{Primitives}
We work on a probability space $(\Omega, \mathcal{F}, \mathbb{P})$, where we suppose that $(\Omega, \mathcal{F}, \mathbb{P})$ has no atoms and that 
$L^1(\Omega, \mathcal{F}, \mathbb{P})$ is separable. For every $p \in [1,\infty]$, we denote by $L^p(\mathbb{P})$ the collection of all real-valued $\mathcal{F}$-measurable random variables with finite $L^p$ norm under $\mathbb{P}$. In particular, 
$L^{\infty}:=L^{\infty}(\mathbb{P})$ is the collection of essentially bounded random variables. 
Consider a finite set $N$ with $|N|=n \geq 2$ of agents endowed with initial risk positions $(X_i)_{i \in N} \in (L^{\infty})^n$. 
Let $U_i: L^{\infty} \to \mathbb{R}$ be a monetary utility function of agent $i \in N$. 
We assume that $U_i$ is concave, cash-invariant, i.e. 
\[U_i(\xi+c)=U_i(\xi)+c \quad \forall \xi \in L^{\infty}, \ \forall c \in \mathbb{R},\] 
and monotone with respect to the order of $L^{\infty}$. 
From the cash invariance and the monotonicity of each $U_i$, we see that $U_i$ is $1$-Lipschitz:
\[
\xi \leq \eta + \|\xi-\eta\|_{\infty} \Rightarrow U_i(\xi) \leq U_i(\eta) + \|\xi-\eta\|_{\infty} \Rightarrow |U_i(\xi)-U_i(\eta)| \leq \|\xi-\eta\|_{\infty}.
\]
We normalise $U_i(0)=0$. 
Cash invariance is essential because our P\&C mechanism uses deterministic transfers; it makes posted prices and bids enter payoffs linearly.  

\begin{remark}[Cash invariance / monetary transfers]
Cash invariance means that utilities concerning cash are expressed in units of a riskless numeraire: adding a sure amount $c$
to a player's payoff shifts her evaluation by $c$. This is the natural assumption when players can make
deterministic transfers (as in the posted prices $p(\xi)$ and auction bids) or have access to a risk-free asset.
It is also standard in risk-sharing and risk-measure foundations, where $U_i$ is interpreted as a certainty
equivalent (in cash units) of a risky payoff.
\end{remark}

Let $\mathbb{A}(X)$ denote the set
\[\mathbb{A}(X):=\left\{(\xi_i)_{i \in N} \in (L^{\infty})^n \bigg\vert \sum_{i \in N} \xi_i =X\right\},\]
where $X := \sum_{i \in N} X_i$ is the total risk. 
We assume that the total risk is nontrivial:
\[
\mathbb{P}(|X| >0) >0.
\]
The set $\mathbb{A}(X)$ consists of attainable risk allocations for the agents.  

Given functions $U_i$, $i \in N$, we denote by 
\[U(X):=U_1\square \cdots \square U_n(X):=\sup_{(\xi_i)_{i \in N} \in \mathbb{A}(X)} \sum_{i \in N} U_i(\xi_i), \quad X \in L^{\infty},\]
the sup-convolution of concave functions $U_i$, $i \in N$, which follows the notation of \cite{jouini2008optimal}. 

\subsection{Feasible allocations}
To consider situations where agents share losses, we focus on the cases where the sign of $\xi_i$ is the same as $X$ for all $i \in N$. More formally, write
\[
\Delta_{X} :=\{
\xi \in \mathbb{A}(X) \, \vert \, 
\xi_i X \geq 0 \ \mathbb{P}\text{-a.s.}, \ \mathrm{1}_{\{X=0\}}\xi_i=0 \ \mathbb{P}\text{-a.s.} \ \forall i \in N
\}.
\]
Throughout the remainder of the paper, let 
\[
K:=\Delta_X.
\]
The same-sign restriction is not intrinsic to the P\&C implementation
argument. Its principal role is to provide an economically interpretable
weak$^\ast$-compact menu of pure loss-sharing allocations. Indeed, under
this restriction each component satisfies
\[
|\xi_i|\leq |X|,
\]
whereas the full attainable set contains arbitrary zero-sum side bets and
is therefore unbounded.

All implementation results continue to hold for any alternative nonempty
weak$^\ast$-compact menu on which the utility functionals satisfy the
stated continuity conditions. For example, one may replace the same-sign
menu by
\[
K^\kappa(X)
=
\{\xi\in\mathbb{A}(X):|\xi_i|\leq\kappa|X|
\text{ a.s. for all }i\}, \quad \kappa \geq 1, 
\]
which permits limited opposite-sign positions when $\kappa>1$.

\subsection{Feasible allocations equipped with weak$^{\ast}$ topology}
Because the economic cost of an allocation may depend not only on its
marginal distribution but also on its joint distribution with background
and secondary stress factors \citep{doherty1983optimal}, we equip
$\Delta_X$ with the weak$^\ast$ topology to retain a broad menu of
allocations that can depend on background factors. The set
$\Delta_X\subset\mathbb A(X)$ is then compact in the weak$^\ast$ topology
$\sigma((L^\infty)^n,(L^1)^n)$.
\begin{lemma}\label{lem:Delta-compact}
Fix $X \in L^{\infty}$. Then, $\Delta_{X} \subset \mathbb{A}(X)$ is compact in the weak$^{\ast}$ topology in $\sigma((L^{\infty})^n, (L^1)^n)$. 
\end{lemma}
\noindent
\emph{Proof: see Appendix \ref{sec:appA}}.
\begin{remark}[Use of a weak$^\ast$-compact menu]
Weak$^\ast$ compactness allows us to retain the entire same-sign
allocation menu $\Delta_X$ while allowing allocations to depend on
background factors. In our setting, using a norm-compact menu requires
additional restrictions on admissible sharing rules and hence a
restriction of $\Delta_X$. In general, such a restricted menu need not
contain an allocation that is optimal over all of $\Delta_X$.

The weak$^\ast$ formulation also accommodates norm-compact menus.
On any such menu, the norm topology and the relative weak$^\ast$
topology coincide \citep[Theorem~2.36, p.~41]{aliprantis2006infinite},
so norm continuity of monetary utilities implies relative weak$^\ast$
continuity on that menu. Thus, models that restrict the menu to be
norm compact while relaxing the assumptions on utilities are included
as special cases of our framework. Appendix~\ref{app:entropic-simplex}
provides an example in which P\&C applies to entropic utility on a
norm-compact menu.
\end{remark}

\subsection{Efficiency notion and welfare benchmark}
\begin{definition}[Pareto optimal allocation] ~ \label{def:pareto-optimal}

Let $(\xi_i)_{i \in N} \in \mathbb{A}(X)$ be an attainable allocation. We say that $(\xi_i)_{i \in N}$ is Pareto optimal if for all $(\zeta_i)_{i \in N} \in \mathbb{A}(X)$:
\[U_i(\zeta_i) \geq U_i(\xi_i) \quad \forall i \Longrightarrow U_i(\zeta_i)=U_i(\xi_i) \quad \forall i. \]
\end{definition}

\begin{theorem}\label{thm:sup-convolution-pareto-optimal}
Let $(U_i)_{i \in N}$ be a sequence of monetary utility functions. For a given aggregate risk $X \in L^{\infty}$ and $(\xi)_{i \in N} \in \mathbb{A}(X)$, the following statements are equivalent:
\begin{itemize}
\item[(i)] $(\xi_i)_{i \in N}$ is a Pareto optimal allocation.
\item[(ii)] $U_1\square \cdots \square U_n(X)=\sum_{i \in N} U_i(\xi_i)$
\end{itemize}
\end{theorem}
This equivalence is proved for two agents by \citet[\,Thm.3.1]{jouini2008optimal}.
For completeness and to keep notation consistent with the present paper, we reproduce a full, self-contained proof for the general 
$n$-agent case in Appendix \ref{sec:appB}.

\subsection{Efficiency notion including monetary transfers}
The feasible set $\Delta_X$, unlike $\mathbb A(X)$, is not closed under
the addition of zero-sum transfers. To recover the closure property used
in Theorem~\ref{thm:sup-convolution-pareto-optimal}, we introduce a notion
of efficiency that allows for monetary transfers.

Let
\[
D_0
:=
\left\{
c\in\mathbb R^n:
\sum_{i=1}^n c_i=0
\right\},
\qquad
\widetilde K:=K+D_0.
\]
The set $K$ describes admissible risk allocations, while
$\widetilde K$ describes final payoff allocations after deterministic
zero-sum transfers.

\begin{definition}[Transfer Pareto efficiency] ~ \label{def:t-Pareto}

We call $\xi\in K$ transfer-Pareto efficient if
$\xi+c$ is Pareto optimal in $\widetilde K$ for some, equivalently every,
$c\in D_0$.
\end{definition}

\begin{proposition}[Transfer-Pareto efficiency] ~ \label{prop:transfer-Pe}
For any $\xi\in K$, the following are equivalent:
\begin{itemize}
\item[(i)] $\xi+c$ is Pareto optimal in $\widetilde K$
      for some, equivalently every, $c\in D_0$;
\item[(ii)] $\xi\in
\arg\max_{\eta\in K}
\sum_{i=1}^n U_i(\eta_i)$.
\end{itemize}
\end{proposition}
\begin{proof}
$(i) \Longrightarrow (ii)$ Suppose (ii) does not hold: there is $\eta \in K$ such that 
\[
s:=\sum_i U_i(\eta_i) - \sum_i U_i(\xi_i) >0.
\]
Choose $r=(r_1,\cdots, r_n)$ with $\sum_i r_i =s$, $r_i \geq 0$ for all $i$, and 
 set $c'=(c'_1,\cdots, c'_n)$ such that
\[
c'_i =U_i(\xi_i) +c_i -U_i(\eta_i) + r_i.
\] Then $\sum_i c'_i =0$ and 
\[
U_i(\eta_i +c'_i) =U_i(\xi_i +c_i) + r_i,
\] and we see $\eta+c'$ Pareto improves $\xi+c$. 

$(ii) \Longrightarrow (i)$ Suppose $\eta + c' \in \widetilde K$ Pareto dominates $\xi +c$. Then we would have from cash invariance and $\sum_i c_i = \sum_i c'_i =0$, 
\[
\sum_i U_i(\eta_i) =\sum_iU_i(\eta_i + c'_i) > \sum_i U_i(\xi_i+c_i) =\sum_i U_i(\xi_i).
\] A contradiction. 
\end{proof}
Proposition~\ref{prop:transfer-Pe} shows that the set of transfer-Pareto-efficient
risk allocations is exactly
\[
\arg\max_{\eta\in K}
\sum_{i=1}^nU_i(\eta_i).
\]

\section{Price-and-Choose mechanism on an infinite menu}\label{sec:P&C}
In this section, we study the P\&C mechanism on an infinite menu.
For this purpose, we introduce normalized weak$^\ast$-continuous price schedules.
We write $C(K)$ for the real-valued functions continuous on $K$ in its
weak$^\ast$ topology.

\begin{definition}
Fix any full-support Borel probability measure $\mu_K$ on $K$. We define
the set of price schedules by
\begin{align}
P_n:=\left\{
 p\in C(K)\ \Big\vert\ \int_K p(\xi)\,d\mu_K(\xi)=0
\right\}.\label{eq:price-schedule}
\end{align}
\end{definition}
A Borel probability measure $\mu$ on $K$ has full support if $\mu(O)>0$
for every nonempty open set $O\subset K$.
Since $L^1(\Omega,\mathcal F,\mathbb P)$ is separable in its norm topology
and $K$ is norm bounded, $K$ is separable in its relative weak$^\ast$ topology. Let
$\{x_m\}_{m\ge1}$ be a countable dense subset of $K$ and define
\[
\mu_K=\sum_{m=1}^{\infty}2^{-m}\delta_{x_m}.
\]
Then $\mu_K$ is a Borel probability measure on $K$. Moreover, it has full
support: for any nonempty open set $O\subset K$, density implies
$x_m\in O$ for some $m$, hence $\mu_K(O)\ge2^{-m}>0$.

The integral condition $\int_Kp\,d\mu_K=0$ is a normalization that fixes the
otherwise arbitrary additive constant in a price schedule. Adding a
constant $c$ to $p$ does not affect the follower's choice, since
\[
\arg\max_{\xi\in K}\{U_i(\xi)-p(\xi)\}
=\arg\max_{\xi\in K}\{U_i(\xi)-p(\xi)-c\}.
\]
It changes the level of transfers, but not the ranking of allocations.
The normalization selects one representative from each equivalence class
$p+c$.

For brevity, we henceforth write $U_i(\xi)$ for $U_i(\xi_i)$,
where $\xi=(\xi_1,\dots,\xi_n)\in K$.

\medskip
\noindent
\begin{assumption}[Weak$^\ast$ continuity]~\label{ass:weakstar-continuity}
For each $i\in N$, the map $\xi\mapsto U_i(\xi)$ is weak$^\ast$-continuous
on $K$. That is, $U_i\in C(K)$.
\end{assumption}
\medskip
\noindent
Continuity on the compact menu ensures that social welfare attains its
maximum on $K$. The following example illustrates an optimal allocation in this
environment.

\begin{example}[Sharing a residual catastrophic layer]
\label{ex:catastrophic-layer}

A syndicate first places the layers of catastrophic loss that can be accepted
by outside markets. The residual loss must be allocated internally among
its three members. Let $S\geq0$ be the aggregate gross loss and let $C(S)$
be the ceded loss, with $0\leq C(S)\leq S$. Write
\[
R:=S-C(S),
\qquad
X:=-R\leq0,
\]
and suppose that $R$ is essentially bounded and nonconstant. The allocation menu remains
$K=\Delta_X$.

In practice, scenario analysis evaluates risk under different stress conditions
associated with background factors and combines the resulting assessments.
\citet{wang2021scenario} formalize a general framework for such scenario-based
risk evaluation. We apply this framework to the sharing of a catastrophic loss
layer.

The members assess allocations using a finite collection of stress-test
distributions $Q_1,\ldots,Q_m$, with $m\geq3$.
For example, $Q_s$ may place greater weight on states in which the cost of
materials needed for reconstruction rises sharply, while $Q_{s'}$ may place
greater weight on states in which secondary disasters occur.
Each scenario $Q_s$ is a countably additive
probability measure dominated by $\mathbb P$, with density
$q_s=dQ_s/d\mathbb P\in L^1(\mathbb P)$. We include the reference model as
$Q_1=\mathbb P$. The scenarios and their probabilities,
\[
\pi_s>0,\qquad \sum_{s=1}^m\pi_s=1
\]
are common across members. Suppose that the models give at least two
different values of $\mathbb E_{Q_s}[X]$, so that the scenario comparisons
are relevant for the retained loss.

For $i=1,2,3$, choose $\tau_i>0$ and define
\[
U_i(Y)
=
-\tau_i\log\left[
\sum_{s=1}^m\pi_s
\exp\left(-\frac{\mathbb E_{Q_s}[Y]}{\tau_i}\right)
\right],
\qquad Y\in L^\infty.
\]
The parameter
$\tau_i$ measures tolerance for variation in scenario expected payoffs.
The exponential transformation is applied to the scenario means, whereas
ordinary entropic utility applies it directly to the random payoff.
Accordingly, this criterion evaluates exposure through the scenario means.

The aggregator is concave and strictly increasing in each scenario mean.
Since the scenario expectations are positive linear functionals, $U_i$ is
concave and monotone. It is strictly monotone because $Q_1=\mathbb P$ and
$\pi_1>0$. Moreover, $\sum_s\pi_s=1$ gives $U_i(0)=0$, and
\[
U_i(Y+c)=U_i(Y)+c,\qquad c\in\mathbb R.
\]
Its concavity is strict between payoffs $Y$ and $Z$ whenever the vector
$(\mathbb E_{Q_s}[Y-Z])_{s=1}^m$ is nonconstant across scenarios. If all
these differences are equal, the concavity inequality is an equality.
Thus strict curvature concerns contrasts between scenario expected payoffs.

For each $s$, since $q_s\in L^1(\mathbb P)$,
\[
\mathbb E_{Q_s}[Y]=\mathbb E_{\mathbb P}[q_sY]
\]
is weak$^\ast$-continuous. Since $U_i$ is a continuous function of finitely
many such expectations, it is weak$^\ast$-continuous on the full menu $K$.

In this example, an explicit optimal allocation follows from the common
scenarios and scenario probabilities. Set
\[
T:=\sum_{i=1}^3\tau_i,
\qquad
w_i:=\frac{\tau_i}{T},
\qquad
\xi_i^\ast:=w_iX.
\]
For any attainable allocation $\xi\in\mathbb A(X)$, put
$z_{is}:=\mathbb E_{Q_s}[\xi_i]$. By linearity of expectation, for every $s$,
\[
\sum_{i=1}^3z_{is}=\mathbb E_{Q_s}[X].
\]
Apply H\"older's inequality on the finite scenario space with probabilities
$\pi_s$ and exponents $1/w_i=T/\tau_i$. Since $\sum_iw_i=1$,
\[
\begin{aligned}
\sum_{s=1}^m\pi_s e^{-\mathbb E_{Q_s}[X]/T}
&=
\sum_{s=1}^m\pi_s
\prod_{i=1}^3
\left(e^{-z_{is}/\tau_i}\right)^{w_i}\\
&\leq
\prod_{i=1}^3
\left(\sum_{s=1}^m\pi_s e^{-z_{is}/\tau_i}\right)^{w_i}.
\end{aligned}
\]
Taking logarithms and multiplying by $-T$ yields
\[
\sum_{i=1}^3U_i(\xi_i)
\leq
-T\log\left[
\sum_{s=1}^m\pi_s e^{-\mathbb E_{Q_s}[X]/T}
\right].
\]
At the proportional allocation,
\[
\frac{\mathbb E_{Q_s}[\xi_i^\ast]}{\tau_i}
=
\frac{\mathbb E_{Q_s}[X]}{T},
\]
so equality holds and
\[
\sum_{i=1}^3U_i(\xi_i^\ast)
=
-T\log\left[
\sum_{s=1}^m\pi_s e^{-\mathbb E_{Q_s}[X]/T}
\right].
\]
Thus $\xi^\ast$ maximizes aggregate monetary utility over $\mathbb A(X)$.
Since $w_i>0$ and $\sum_iw_i=1$, it belongs to the full same-sign menu
$\Delta_X$ and also maximizes welfare there.

Section~\ref{sec:credal-set} introduces a general framework in which
credal sets, probability measures on these sets, and prior-specific
utilities may differ across members.
Example~\ref{ex:finite-scenario-smooth-ambiguity} extends this specification
to member-specific scenario sets and probabilities, including the present
model as the special case of common scenarios and probability weights.
\end{example}

\subsection{Price-and-choose risk sharing with two players}
In this subsection, we set $n=2$, so the feasible set is $K=\Delta_{X}$. 
We assume that players 1 and 2 commit to using the Price and Choose mechanism adapted from \cite{echenique2025price}. 
Under the mechanism, player 1 sets a price function $p \in P_2$, which is a weak$^\ast$-continuous function on $K$. 
We assume that a price $p(\xi)$ represents the amount that player 2 pays to player 1 if player 2 chooses the allocation $\xi$ from their choice set $K$.  
Prices may be positive or negative. The schedule is normalized so that
its integral over $K$ with respect to $\mu_K$ is zero.
Timing is given as follows. 
\begin{itemize}
\item[1.] Player 1 sets $p \in P_2$, where $P_2$ is defined in \eqref{eq:price-schedule}. 
\item[2.] Player 2 chooses $\xi \in K$ and pays $p(\xi)$ to player 1.
\end{itemize}
The payoff according to the mechanism is given by
\[g(p,\xi) =(g_1(p,\xi), g_2(p,\xi)) =(U_1(\xi)+p(\xi), U_2(\xi)-p(\xi)).\]
The mechanism provides an extensive form game, where a strategy profile $\sigma=(\sigma_1, \sigma_2)$ is given by 
$\sigma_1 \in P_2$ and $\sigma_2: P_2 \to K$. 
Throughout the paper, SPNE refers to a pure-strategy subgame-perfect Nash equilibrium. 
Following \cite{echenique2025price}, we say that a P\&C mechanism
fully implements the set of transfer-Pareto-efficient allocations
in SPNE if the following two conditions hold.
\begin{itemize}
\item[1.] In every SPNE $\sigma=(\sigma_1,\sigma_2)$, the allocation
$\sigma_2(\sigma_1)$ is transfer-Pareto efficient in the sense of
Definition~\ref{def:t-Pareto}.
\item[2.] For every transfer-Pareto-efficient risk-sharing allocation
$\xi$, there exists an SPNE $\sigma=(\sigma_1,\sigma_2)$ such that
$\xi=\sigma_2(\sigma_1)$.
\end{itemize}
\begin{theorem}[Two-player P\&C implements efficiency on $K$]\label{thm:2-player} ~

\begin{itemize}
\item[(i)] In any SPNE, the chosen allocation $\xi \in K$ maximizes $U_1(\xi)+U_2(\xi)$.
\item[(ii)] For every transfer-Pareto-efficient allocation
\[
\xi\in
\arg\max_{\eta\in K}
\{U_1(\eta)+U_2(\eta)\},
\]
there exists an SPNE implementing $\xi$.
%
\end{itemize}
\end{theorem}
To prove this theorem, we use the following lemma. 

\begin{lemma}\label{lem:closed-set-A_2}
Fix $p\in P_2$. Let $K:=\Delta_{X}$ endowed with the weak$^{\ast}$ topology $\sigma((L^\infty)^2,(L^1)^2)$.
Then the best–response set
\[
\mathbb{A}_p^2\ :=\ \arg\max_{\xi\in K}\,\big\{\,U_2(\xi)-p(\xi)\,\big\}
\]
is nonempty and compact (in the weak$^{\ast}$ topology on $K$).
\end{lemma}
\noindent
\emph{Proof: see Appendix \ref{sec:appD}}.

\begin{proof}[\it Proof of Theorem \ref{thm:2-player}]

Step 1: 
We begin by proving the existence of a price schedule that makes player 2 indifferent between all options, and the uniqueness of the price schedule. 

Define $\mathrm{Avg}_2:=\int_K U_2(\xi)\,d\mu_{K}(\xi)$ and set
\[
p^*(\xi):=U_2(\xi)-\mathrm{Avg}_2,\qquad \xi\in K.
\]
Then $U_2(\xi)-p^*(\xi)\equiv \mathrm{Avg}_2$, so player 2 is indifferent among all
$\xi\in K$ under $p^*$.
Moreover, $p^*$ is weak$^\ast$-continuous and
$\int_K p^*\,d\mu_K=\int_K U_2\,d\mu_K-\mathrm{Avg}_2=0$,
hence $p^*\in P_2$.

Uniqueness: if $p\in P_2$ makes player 2 indifferent, i.e. $U_2(\xi)-p(\xi)\equiv \theta$ on $K$,
then, since $\mu_{K}$ is a probability measure,
\[
\theta=\int_K \theta\,d\mu_{K}
=\int_K (U_2-p)\,d\mu_{K}
=\int_K U_2\,d\mu_{K}-\underbrace{\int_K p\,d\mu_{K}}_{=0}
=\mathrm{Avg}_2,
\]
so $p(\xi)=U_2(\xi)-\mathrm{Avg}_2=p^*(\xi)$ for all $\xi\in K$.
 
Step 2: 
 Next, we claim that if $\sigma$ is a SPNE, then 
 \begin{align*}
&\sigma_1=p^*\\ 
& \sigma_2(p^*) \in \arg\max_{\xi \in K} W(\xi), \ W(\xi):=U_1(\xi)+U_2(\xi).
 \end{align*}
 
For any $p \in P_2$, define 
\[
v_2(p):=\max_{\xi \in K} (U_2(\xi)-p(\xi)).
\]
Then we have 
\[
v_2(p) \geq \int_K (U_2(\xi)-p(\xi)) d \mu_{K}=\mathrm{Avg}_2.
\]
Choose $\xi_p \in \mathbb{A}_p^2$. Player 1's payoff at $(p,\xi_p)$ is given by 
\[
U_1(\xi_p) + p(\xi_p) = U_1(\xi_p) + U_2(\xi_p) -(U_2(\xi_p)-p(\xi_p)) = W(\xi_p) -v_2(p).
\] 
Hence, 
\[
U_1(\xi_p) + p(\xi_p) \leq W^{\max} -\mathrm{Avg}_2, \quad W^{\max}:=\max_{\xi \in K} W(\xi)
\]
So in any SPNE player 1's on-path payoff cannot exceed $W^{\max} -\mathrm{Avg}_2$. 

Define
\[
\overline W:=\int_KW(\xi)\,d\mu_K(\xi),
\qquad D:=W^{\max}-\overline W\ge0.
\]
For $\varepsilon>0$, consider the price schedule
\[
p^\varepsilon(\xi)
:=p^*(\xi)-\varepsilon\bigl(W(\xi)-\overline W\bigr).
\]
Since $p^*$ and $W$ are weak$^\ast$-continuous and
\[
\int_Kp^\varepsilon\,d\mu_K
=\int_Kp^*\,d\mu_K
-\varepsilon\left(\int_KW\,d\mu_K-\overline W\right)=0,
\]
we have $p^\varepsilon\in P_2$ for every $\varepsilon>0$.
Under this price, player 2's payoff is
\[
U_2(\xi)-p^\varepsilon(\xi)
=\mathrm{Avg}_2+\varepsilon\bigl(W(\xi)-\overline W\bigr).
\]
Consequently,
\[
\arg\max_{\xi\in K}\{U_2(\xi)-p^\varepsilon(\xi)\}
=\arg\max_{\xi\in K}W(\xi).
\]
Every best response therefore maximizes welfare, and at any such response
player 1 receives
\[
U_1(\xi)+p^\varepsilon(\xi)
=W^{\max}-\mathrm{Avg}_2-\varepsilon D.
\]

Now let $\sigma$ be a SPNE with on-path payoff $\Pi_\sigma$ for player 1.
For any $\delta>0$, choose $0<\varepsilon<\delta/(1+D)$. Then deviating
to $p^\varepsilon$ gives player 1
\begin{align}
W^{\max}-\mathrm{Avg}_2-\varepsilon D
>W^{\max}-\mathrm{Avg}_2-\delta.\label{eq:deviation}
\end{align}
If $\Pi_{\sigma}<W^{\max} -\mathrm{Avg}_2$, choose $\delta=(W^{\max} -\mathrm{Avg}_2-\Pi_{\sigma})/2>0$. Then \eqref{eq:deviation} gives a deviation yielding payoff strictly greater than $\Pi_{\sigma}$. A contradiction. 
Hence, any SPNE must satisfy
\begin{align}
\Pi_{\sigma} = W^{\max} -\mathrm{Avg}_2. \label{eq:spne-player1-payoff}
\end{align}
Let $p=\sigma_1$ and $\xi=\sigma_2(p) \in \arg\max (U_2-p)$. Then
\[
\Pi_{\sigma}=U_1(\xi)+p(\xi)=W(\xi)-v_2(p).
\]
Since $W(\xi) \leq W^{\max}$ and $v_2(p) \geq \mathrm{Avg}_2$,
\eqref{eq:spne-player1-payoff} forces both inequalities to hold with equality:
\[
W(\xi) = W^{\max}, \ v_2(p)=\mathrm{Avg}_2.
\] 
Let $f:=U_2-p\in C(K)$. Then
\[
\max_{\xi\in K} f(\xi)=v_2(p)=\mathrm{Avg}_2
=\int_K (U_2-p)\,d\mu_{K}=\int_K f\,d\mu_{K}.
\]
Since $\mu_{K}$ has full support and $f$ is continuous, $\int_K f\,d\mu_{K}=\max_K f$
implies that $f$ is constant on $K$. Hence $U_2(\xi)-p(\xi)\equiv \mathrm{Avg}_2$ and therefore
$p=p^*$.
Consequently $\sigma_1=p^*$ and, since $W(\xi)=W^{\max}$ with $\xi=\sigma_2(p)$ and $p=p^*$,
we have $\sigma_2(p^*)\in\arg\max_{\xi\in K} W(\xi)$.

Step 3: 
The remaining task is now to show that  for every efficient outcome, there is a corresponding SPNE. 
Fix any efficient $\xi \in \arg\max W$. Define a strategy profile corresponding to $\xi$ such that 
\[
\sigma_1:=p^{\ast}, \quad \sigma_2(p^{\ast}):=\xi, \quad \sigma_2(p) \in \arg\max(U_2-p) \ \forall p.
\] 
By construction, $\sigma_2(p)$ is a best response after every $p \in P_2$. 

For player 1: for any deviation to $p$, the induced outcome is $\eta=\sigma_2(p)$, and the corresponding payoff is 
\[
W(\eta)-v_2(p) \leq W^{\max}-\mathrm{Avg}_2.
\] Under $(p^{\ast}, \xi)$ the player 1 gets $W^{\max}-\mathrm{Avg}_2$. Hence, the player 1 cannot gain. 
This completes the proof.  
\end{proof}

When the two-player P\&C mechanism is viewed as a Stackelberg game,
the follower's best-response correspondence has nonempty compact values
and is upper hemicontinuous when $P_2$ carries the uniform-norm topology.
This follows because $K$ is compact and
$F(\xi,p):=U_2(\xi)-p(\xi)$ is jointly continuous in allocations and prices
\citep[Theorem~17.31]{aliprantis2006infinite}.
However, this correspondence need not be lower hemicontinuous.
Moreover, because admissible price schedules are general continuous
functions, the best-response correspondence need not be convex-valued.
Thus, a conventional approach based on continuous selection, such as an
application of Michael's selection theorem, is unavailable without
additional assumptions \citep{michael1956selected}.
A widely used approach specifies tie-breaking in favor of or against the
leader. An existence result for an SPNE under such a convention does not,
by itself, rule out inefficient SPNE under other selections.

Under pessimistic tie-breaking, where an indifferent follower selects a
best response that is worst for the leader, the leader's value function
may be discontinuous and its supremum need not be attained
\citep[pp.~447, 449]{von2010leadership}.
The proof of Theorem~\ref{thm:2-player} uses small adjustments to the
equalizing price schedule to discipline equilibrium selection and
characterize SPNE allocations directly, without imposing an exogenous
tie-breaking rule.

\subsection{Price-and-choose risk sharing with many players}
We now extend the two-player P\&C mechanism inductively to $n$ players. 
Following \cite{echenique2025price}, we call this the
P$^{n-1}$\&C mechanism. Player 1 first posts a price schedule
$p^2\in P_n$. Each player $i=2,\ldots,n-1$ then posts
$p^{i+1}\in P_n$ after observing the previously posted schedules
$p^2,\ldots,p^i$. Finally, player $n$ chooses a risk-sharing allocation
$\xi\in K$. We write $\sigma_n(p)$ for this choice, where
$p:=(p^2,\ldots,p^n)$ is the complete price history.
The process leads to the following payoffs. 
\begin{align*}
g_1(p,\xi) &= U_1(\xi) +p^2(\xi)\\
g_m(p,\xi) &=U_m (\xi) -p^{m}(\xi)+p^{m+1}(\xi) \quad \text{ for }m=2,\cdots, n-1 \\
g_n(p,\xi) &=U_n(\xi) -p^n(\xi).
\end{align*}
Set $p^{n+1}(\xi)\equiv p^1(\xi)\equiv 0$ on $K$. 
For each $i=1, \cdots, n$, define
\[
W_i(\xi):=\sum_{j=i}^n U_j(\xi), \quad \overline{W}_i:=\int_K W_i(\xi) d\mu_{K}(\xi). 
\]

\begin{theorem}\label{thm:many-player}
For every subgame beginning at stage $i$, the following statements hold. 
 \begin{itemize}
 \item[(A)] The terminal allocation satisfies 
\[
\sigma_n(p) \in \arg\max_{\xi \in K} \{W_i(\xi)-p^i(\xi)\}
\]
 for $i=1,\cdots, n$ in every SPNE; 
 \item[(B)] the equilibrium posted price at stage $i+1$ equalizes continuation welfare: 
\[
p^{i+1}(\xi)=W_{i+1}(\xi)-\overline{W}_{i+1} \quad \forall \xi \in K
\]
holds for $i=1,\cdots, n-1$ in every SPNE; 
\item[(C)] for every
\[
\xi^\ast\in
\arg\max_{\xi\in K}
\{W_i(\xi)-p^i(\xi)\},
\]
there exists an SPNE of the subgame starting at stage $i$
whose terminal allocation is $\xi^\ast$.
 \end{itemize}
Consequently, the P$^{n-1}$\&C mechanism fully implements
\[
\arg\max_{\xi\in K}
\sum_{i=1}^nU_i(\xi_i).
\]
\end{theorem}

\begin{proof}
We first show the case of $i=n-1$. This is a two-player risk sharing between $n-1$ and $n$, with follower utility $U_n$, and by Theorem \ref{thm:2-player}, (A) and (B) hold.

We now proceed by induction. Assume that (A) and (B) hold for any subgame starting from $i+1$ or larger. Then for any given $p^{i+1}$, the continuation outcome 
$\sigma_n(p)$ maximizes $W_{i+1}(\xi)-p^{i+1}(\xi)$ over $\xi \in K$. Therefore, from player $i$'s viewpoint, stage $i$ reduces to a 2-player P\&C problem with leader's payoff $U_i(\xi)-p^i(\xi)+p^{i+1}(\xi)$ and follower's payoff $W_{i+1}(\xi)-p^{i+1}(\xi)$, where the leader posts $p^{i+1}$ and the follower chooses $\xi$. 

Since $W_{i+1}$ is weak$^\ast$-continuous, its normalized schedule
$W_{i+1}-\overline W_{i+1}$ belongs to $P_n$.
Applying the proof of Theorem~\ref{thm:2-player} to this reduced game with
leader payoff $U_i-p^i$ and follower utility $W_{i+1}$, we obtain the
unique equalizing price
\[
p^{i+1}(\xi) = W_{i+1}(\xi)-\overline{W}_{i+1} \quad \forall \xi \in K,
\]
and the implemented allocation maximizes the sum of leader and follower payoffs: 
\[
(U_i-p^i+p^{i+1})+(W_{i+1}-p^{i+1})=W_i-p^i.
\]
Hence, (A) and (B) hold for subgames starting from $i$. This completes the induction, and taking $i=1$ with $p^1 \equiv 0$ gives efficiency.  

We proceed to show (C). For the case $i=n-1$, statement (C) follows from Theorem \ref{thm:2-player}. 
Suppose that the statement (C) holds for every subgame beginning at stage $i+1$.
For each $q\in P_n$, let $\mathcal O_{i+1}(q)$ 
denote the set of terminal allocations induced by continuation
SPNEs beginning at stage $i+1$.
By (A) and the induction hypothesis (C),
\[
\mathcal O_{i+1}(q)
=
\arg\max_{\xi\in K}
\{W_{i+1}(\xi)-q(\xi)\}.
\]
Hence, the subgame beginning at stage $i$ is equivalent to the
two-player P\&C game in which the reduced follower's response set
to $q$ is $\mathcal O_{i+1}(q)$.
Set 
\[
p^{i+1,\ast}(\xi)=W_{i+1}(\xi)-\overline{W}_{i+1}, \quad \forall \xi \in K.
\]
But 
\[
\mathcal O_{i+1}(p^{i+1, \ast})=K, 
\] since $W_{i+1}(\xi)-p^{i+1,\ast}(\xi)$ is constant. This gives 
existence of continuation SPNE starting from stage $i+1$ with the terminal allocation $\xi^*$. 

By the full-implementation argument of Theorem~\ref{thm:2-player} applied to the
reduced game, player \(i\) does not profitably deviate from
\(p^{i+1,\ast}\).
This completes the proof. 
\end{proof}

\begin{example}[P\&C for scenario-sensitive sharing of a residual catastrophic layer]
\label{ex:catastrophic-layer-pc}

Consider the three members, common scenario distributions $Q_s$, weights
$\pi_s$, and utilities from Example~\ref{ex:catastrophic-layer}.
Let
\[
T=\tau_1+\tau_2+\tau_3,
\qquad
w_i=\frac{\tau_i}{T}.
\]
The proportional allocation $\xi_i^\ast=w_iX$ maximizes aggregate
monetary utility on the full menu $K=\Delta_X$.

Fix the full-support normalization measure $\mu_K$ used by the mechanism
and define
\[
\mathrm{Avg}_i:=\int_KU_i(\eta_i)\,d\mu_K(\eta),
\qquad i=2,3,
\]
and
\[
W_2(\xi):=U_2(\xi_2)+U_3(\xi_3),
\qquad
\overline W_2:=\int_KW_2(\eta)\,d\mu_K(\eta)
=\mathrm{Avg}_2+\mathrm{Avg}_3.
\]
In the three-player P\&C mechanism, member 1 posts $p^2$, member 2 observes
$p^2$ and posts $p^3$, and member 3 chooses the allocation. The equalizing
schedules are
\[
p^3(\xi)=U_3(\xi_3)-\mathrm{Avg}_3,
\qquad
p^2(\xi)=W_2(\xi)-\overline W_2.
\]
Example~\ref{ex:catastrophic-layer} establishes continuity of the utilities
on $K$. Hence these price schedules are weak$^\ast$-continuous, and their
integrals under $\mu_K$ are zero. They are therefore admissible elements
of $P_3$.

If $\xi\in K$ is selected under these schedules, the members' payoffs are
\begin{align*}
U_1(\xi_1)+p^2(\xi)
&=\sum_{i=1}^3U_i(\xi_i)-\mathrm{Avg}_2-\mathrm{Avg}_3,\\
U_2(\xi_2)-p^2(\xi)+p^3(\xi)
&=\mathrm{Avg}_2,\\
U_3(\xi_3)-p^3(\xi)
&=\mathrm{Avg}_3.
\end{align*}
Thus $p^3$ makes member 3 indifferent over the menu, while member 1's
payoff increases with aggregate welfare.

The full-implementation construction in
Theorem~\ref{thm:many-player} supplies continuation strategies that,
together with the equalizing schedules, support
\[
\xi^\ast=(w_1X,w_2X,w_3X)
\]
as the terminal allocation of an SPNE. Conversely, every SPNE terminal
allocation maximizes aggregate utility on $K$. By
Proposition~\ref{prop:transfer-Pe}, these outcomes are transfer-Pareto
efficient relative to $K+D_0$. P\&C therefore implements efficient
scenario-sensitive sharing of the retained catastrophic loss on the full
weak$^\ast$-compact menu.
With common knowledge of all payoff-relevant primitives, this
construction illustrates Theorem~\ref{thm:max-min}.
\end{example}

\section{Multiple prior extension}\label{sec:credal-set}
When players have diverse probability assessments of the uncertain events affecting the payoff, does P\&C still implement an efficient allocation? 
We extend the analysis to this scenario by allowing each player $i$ to evaluate allocations under a credal set 
$\mathcal{P}_i$ of priors over $(\Omega, \mathcal{F})$. We allow priors
to be finitely additive probability charges. Each player also holds a
countably additive belief measure $\Lambda_i$ over $\mathcal P_i$.
Both the credal sets and the belief measures may differ across players.
We additionally require every prior to be dominated
by the reference probability $\mathbb P$. This domination condition has a
separate role: it ensures that prior-specific expectations and certainty
equivalents are well defined on $L^\infty(\mathbb P)$ and that all players
evaluate the same $\mathbb P$-a.s.-defined feasible menu. 

Motivated by the smooth-ambiguity approach of \citet{Klibanoff2005},
we aggregate prior-specific monetary utilities through an exponential
certainty equivalent under $\Lambda_i$. This aggregation preserves
cash invariance, so deterministic transfers continue to enter utility
additively.
Working in the $ba$ space introduced in the next subsection allows us to obtain weak$^\ast$-compact credal sets from the closedness conditions imposed there.
Compactness, together with joint continuity of the prior-specific
utilities in allocations and priors, gives uniform control over their
evaluations near each allocation. It therefore ensures that the
aggregated utilities are weak$^\ast$-continuous on $K$, as required by
the implementation argument.

\subsection{Credal sets}\label{sec:credal-sets}
Let $B_b(\Omega, \mathcal{F})$ denote the Banach space of all bounded $\mathcal{F}$-measurable functions on $\Omega$,
equipped with the sup norm $\|\cdot\|_{\infty}$. 
Its (topological) dual is 
$ba(\Omega,\mathcal{F})$, the space of bounded finitely additive signed measures on $\mathcal{F}$ \citep[][\,ch.4.7]{rao1983theory}.\footnote{
In chapter 4.7 of \cite{rao1983theory}, $\mathcal{C}(\Omega, \mathcal{F})^*=  ba(\Omega,\mathcal{F})$, 
where 
$\mathcal C(\Omega,\mathcal F)$ denotes the Banach space
of bounded $\mathcal{F}$-continuous functions (i.e., functions whose oscillation can be made arbitrarily
small on some finite $\mathcal{F}$-partition of $\Omega$), endowed with the sup norm. 
However, since $\mathcal{F}$ is a $\sigma$-algebra, every bounded $\mathcal{F}$-measurable function can be uniformly approximated by $\mathcal{F}$-simple functions; hence $B_b(\Omega, \mathcal{F}) =  \mathcal{C}(\Omega, \mathcal{F})$. 
}
We equip $ba(\Omega, \mathcal{F})$ with the weak$^{\ast}$ topology $\sigma(ba, B_b)$. 
Under $\sigma(ba, B_b)$, $\mu^n \to \mu$ iff $\int f d\mu^n \to \int f d\mu$ for every bounded $\mathcal{F}$-measurable function $f$.

Let
\[
\mathcal D_{\mathbb P}:=\Big\{
\mu \in ba(\Omega,\mathcal F) \,\Big\vert\,
\mu\geq0,\ \mu(\Omega)=1,\ \mu\ll\mathbb P
\Big\}.
\]
Here $\mu\ll\mathbb P$ denotes the domination condition that
$\mu(A)=0$ whenever $\mathbb P(A)=0$.
For each $i\in N$, let $\mathcal P_i$ be a nonempty
$\sigma(ba,B_b)$-closed subset of $\mathcal D_{\mathbb P}$.
The sets $\mathcal P_i$ may differ across players.

Equip $\mathcal P_i$ with the relative topology induced by the weak$^\ast$
topology $\sigma(ba,B_b)$ on $ba$. Let $\mathcal B(\mathcal P_i)$ denote
the Borel $\sigma$-algebra generated by this relative topology.
For each $i\in N$, fix a countably additive
probability measure
\[
\Lambda_i:\mathcal B(\mathcal P_i)\longrightarrow[0,1].
\]
Each $\nu\in\mathcal P_i$ is a candidate probability measure on
$(\Omega,\mathcal F)$, while $\Lambda_i$ describes player $i$'s beliefs
over these candidate priors. 
Both $\mathcal P_i$ and $\Lambda_i$
may differ across players, and $\mathcal P_i$ need not be finite.

\begin{remark}
By domination, $\sum_j \xi_j=X$ $\mathbb P$-a.s.\ implies $\sum_j \xi_j=X$ $\mu$-a.s.\ for any $\mu\in\mathcal P_i$, so all players share the same feasible set $K$.
\end{remark}

For each $\mu\in\mathcal P_i$, define an expectation operator on
$L^\infty(\mathbb P)$ as follows. Given $Y\in L^\infty(\mathbb P)$,
choose any bounded $\mathcal F$-measurable representative
$f\in B_b(\Omega,\mathcal F)$ of $Y$ and set
\[
\mathbb E_\mu[Y]:=\int_\Omega f\,d\mu
=\langle f,\mu\rangle_{(B_b,ba)}.
\]
If $g\in B_b(\Omega,\mathcal F)$ is another representative of $Y$, then
$\mathbb P(f\ne g)=0$. Since $\mu\ll\mathbb P$, we have
$\mu(f\ne g)=0$, and hence $\int f\,d\mu=\int g\,d\mu$.
Thus the value is independent of the chosen representative, and
\[
\mathbb E_\mu:L^\infty(\mathbb P)\longrightarrow\mathbb R
\]
is well defined. This operator is linear and monotone and, since
$\mu\geq0$ and $\mu(\Omega)=1$, satisfies
\[
|\mathbb E_\mu[Y]|\leq\|Y\|_\infty,
\qquad
\mathbb E_\mu[1]=1.
\]
In particular, it is norm-continuous and belongs to
$(L^\infty(\mathbb P))^*$.

\begin{lemma}\label{lem:credal-set-compactness}
Each credal set $\mathcal{P}_i$ is weak$^{\ast}$-compact.
\end{lemma}
\noindent
\emph{Proof: see Appendix \ref{sec:appE}}.

\subsection{Exponential smooth-ambiguity monetary utilities}\label{sec:max-min-utilities}
For every player $i\in N$ and prior $\nu\in\mathcal P_i$, let
\[
V_{i,\nu}:L^\infty(\mathbb P)\longrightarrow\mathbb R
\]
be a concave, monotone, cash-invariant monetary utility,
normalized so that \mbox{$V_{i,\nu}(0)=0$}.
For every fixed $Y\in L^\infty$, we require the map $\nu\longmapsto V_{i,\nu}(Y)$, $\nu \in \mathcal{P}_i$ 
to be $\mathcal B(\mathcal P_i)$-measurable, and the map
\[
F_i:K\times\mathcal P_i\longrightarrow\mathbb R,
\qquad F_i(\xi,\nu):=V_{i,\nu}(\xi_i),
\]
to be jointly continuous for the product of the relative weak$^\ast$
topology $\sigma((L^\infty)^n,(L^1)^n)$ on $K$ and the relative
weak$^\ast$ topology $\sigma(ba,B_b)$ on $\mathcal P_i$.
Fix a payoff $Y=\xi_i$ that is a component of an allocation $\xi\in K$.
The joint continuity of $F_i$ implies that the map
$\nu\mapsto V_{i,\nu}(Y)$ is continuous on $\mathcal P_i$ and hence
Borel measurable. The separate measurability requirement above ensures that
the integral over $\mathcal{P}_i$ is also well defined for every $Y \in L^\infty$.

Fix $\tau_i>0$ and define player $i$'s exponential smooth-ambiguity
utility on $L^\infty$ by
\[
U_i(Y):=
-\tau_i\log\left[
\int_{\mathcal P_i}
\exp\left(-\frac{V_{i,\nu}(Y)}{\tau_i}\right)
\,\Lambda_i(d\nu)
\right],
\qquad Y\in L^\infty.
\]
Normalization, monotonicity, and cash invariance of the prior-specific
utilities give the uniform bound
\[
|V_{i,\nu}(Y)|\leq\|Y\|_\infty
\qquad\text{for every }\nu\in\mathcal P_i.
\]
Together with measurability, this bound ensures that the integral is
finite and strictly positive for every $Y\in L^\infty$. Thus $U_i$
is well defined and real valued on $L^\infty$.


When the priors are countably additive and
$V_{i,\nu}(Y)=\mathbb E_\nu[Y]$, this criterion is the monetary
certainty-equivalent representation of the smooth-ambiguity model of
\citet{Klibanoff2005}, with exponential ambiguity aggregation
$\phi_i(z)=-\exp(-z/\tau_i)$. The logarithmic transformation applies
$\phi_i^{-1}$ to express the aggregate evaluation in monetary units.
The general prior-specific monetary utilities
allowed here extend this aggregation structure to a broader class of
evaluations.

\begin{lemma}[Properties of exponential smooth-ambiguity monetary utilities]\label{lem:max-min-properties}
Each $U_i$ is real valued, normalized, concave, monotone with respect
to the order of $L^\infty(\mathbb P)$, and cash invariant.
For every $Y\in L^\infty$,
\[
|U_i(Y)|\leq\|Y\|_\infty.
\]
Moreover, $\xi\mapsto U_i(\xi_i)$ is weak$^\ast$ continuous on $K$.
\end{lemma}
\noindent
\emph{Proof: see Appendix \ref{sec:appF}}.

\begin{example}[Finite-scenario specification with heterogeneous assessments]
\label{ex:finite-scenario-smooth-ambiguity}
Retain the three-member syndicate, aggregate payoff $X=-R$, and menu
$K=\Delta_X$ of Example~\ref{ex:catastrophic-layer}, but allow the
members' scenario assessments to differ. For each $i=1,2,3$, let
$m_i\geq1$ be an integer and let $Q_{si}$, $s=1,\ldots,m_i$, be
countably additive probability measures dominated by $\mathbb P$.
Choose $\tau_i>0$ and weights
\[
\pi_{si}>0,\qquad \sum_{s=1}^{m_i}\pi_{si}=1,
\]
and set
\[
\mathcal P_i=\{Q_{si}:s=1,\ldots,m_i\},
\qquad
\Lambda_i=\sum_{s=1}^{m_i}\pi_{si}\delta_{Q_{si}},
\]
where $\delta_{Q_{si}}$ places unit mass at the prior $Q_{si}$.
The scenario counts, scenario measures, and weights may all differ
across members. Thus $\mathcal P_i$ is a member-specific credal set,
and $\Lambda_i$ is a countably additive probability measure over it.
Define
\[
V_{i,\nu}(Y):=\mathbb E_\nu[Y],
\qquad \nu\in\mathcal P_i,\quad Y\in L^\infty.
\]

Each $V_{i,\nu}$ is normalized, linear and hence concave, monotone,
and cash invariant. Since $\mathcal P_i$ is finite, its relative
weak$^\ast$ topology is discrete. Hence, for every $Y\in L^\infty$,
the map $\nu\mapsto V_{i,\nu}(Y)$ on $\mathcal P_i$ is Borel measurable.
For each $\nu=Q_{si}$,
\[
V_{i,Q_{si}}(Y)=\mathbb E_{\mathbb P}[q_{si}Y],
\qquad q_{si}=\frac{dQ_{si}}{d\mathbb P}\in L^1(\mathbb P),
\]
which is weak$^\ast$-continuous in $Y$. Since $\mathcal P_i$ is finite
and discrete, the map $(\xi,\nu)\mapsto V_{i,\nu}(\xi_i)$ is also
jointly continuous on $K\times\mathcal P_i$.
Hence the specification satisfies the conditions of
Subsections~\ref{sec:credal-sets} and~\ref{sec:max-min-utilities},
and Lemma~\ref{lem:max-min-properties} applies.

Substitution into the exponential aggregation formula gives
\[
U_i(Y)=
-\tau_i\log\left[
\sum_{s=1}^{m_i}\pi_{si}
\exp\left(-\frac{\mathbb E_{Q_{si}}[Y]}{\tau_i}\right)
\right],
\qquad Y\in L^\infty.
\]
These utilities are weak$^\ast$-continuous on $K$, so compactness of
$K$ ensures that
\[
\arg\max_{\xi\in K}\sum_{i=1}^3U_i(\xi_i)\ne\varnothing.
\]
Under common knowledge of all payoff-relevant primitives and commitment
to P\&C and the menu $K$, Theorem~\ref{thm:max-min} below establishes
that P\&C fully implements this maximizing set in subgame-perfect
equilibrium.
The heterogeneous specification does not, in general, guarantee the
optimality of the proportional allocation in
Example~\ref{ex:catastrophic-layer}: its H\"older argument relies on
common scenario measures and weights.

\begingroup
\clubpenalty=10000
\widowpenalty=10000
In particular, suppose that, for every member $i$,
\[
m_i=m,\qquad Q_{si}=Q_s,\qquad \pi_{si}=\pi_s,
\quad s=1,\ldots,m,
\]
with the scenario measures and weights of
Example~\ref{ex:catastrophic-layer}. Then the credal sets and the
probability measures over them are common to all members, and the
utilities reduce to those in that example. Consequently, the
proportional allocation
\[
\xi_i^\ast=\frac{\tau_i}{T}X,
\qquad T=\tau_1+\tau_2+\tau_3,
\]
maximizes aggregate monetary utility, and
Example~\ref{ex:catastrophic-layer-pc} gives the equalizing price
schedules and an SPNE implementing it.
\par
\endgroup
\end{example}

\subsection{Implementation under ambiguity}
We assume that players commit to using the P\&C mechanism and the choice
set $K$. Each player evaluates allocations using the exponential
smooth-ambiguity utility defined in
Subsection~\ref{sec:max-min-utilities}. The procedure is that described
in Section~\ref{sec:P&C}.
After observing the previously posted prices, each player
$i=1,\ldots,n-1$ posts a price schedule $p^{i+1}\in P_n$.
The last player, $n$, then chooses a risk-sharing allocation $\xi\in K$.
Set $p^1\equiv p^{n+1}\equiv0$, and write
$p=(p^1,\ldots,p^{n+1})$. Player $k$'s payoff is
\[
g_k(p,\xi) = U_k(\xi_k) -p^k(\xi) +p^{k+1}(\xi), \quad k=1,\cdots, n.
\]
\begin{theorem}[Implementation under exponential smooth ambiguity]~ \label{thm:max-min}

Under the credal-set and belief-measure conditions in
Subsection~\ref{sec:credal-sets} and the exponential smooth-ambiguity
utility specification in Subsection~\ref{sec:max-min-utilities},
assume common knowledge of all payoff-relevant primitives, including
each player's credal set $\mathcal P_i$, belief measure $\Lambda_i$,
parameter $\tau_i$, and prior-specific utilities
$(V_{i,\nu})_{\nu\in\mathcal P_i}$. Then
\begin{enumerate}
\item every SPNE terminal allocation belongs to
\[
\arg\max_{\xi\in K}\sum_{i=1}^n U_i(\xi_i);
\]
\item every allocation in this argmax set is the terminal outcome of some SPNE.
\end{enumerate}
\end{theorem}
\begin{proof}
By Lemma~\ref{lem:max-min-properties}, each $U_i$ is a monetary utility,
and $\xi\mapsto U_i(\xi_i)$ is weak$^\ast$ continuous on $K$.
Thus Assumption~\ref{ass:weakstar-continuity} holds, and the continuation
welfare functions $W_i$ are weak$^\ast$ continuous on $K$.
The follower's best responses are therefore attained on the compact
menu $K$.
The equalizing-price construction,
\[
p^{i+1}(\xi)=W_{i+1}(\xi) - \overline{W}_{i+1} \quad \forall \xi \in K, 
\]
lies in $P_n$ because it is weak$^\ast$ continuous and
\[
\int_K p^{i+1} d\mu_{K}=\int_K (W_{i+1}-\overline W_{i+1})d\mu_{K}=0.
\]
Thus the hypotheses of Theorem~\ref{thm:many-player} are satisfied,
and its backward-induction argument gives both conclusions.
\end{proof}

The set of  implemented optimal allocations is independent of the normalization measure
$\mu_K$, whereas equilibrium payoff levels depend on $\mu_K$.
Thus, $\mu_K$ is a transparent design parameter of the mechanism.

\section{Equal sharing of ex-ante surplus}\label{sec:bidding}

The P$^{n-1}$\&C mechanism favours the first mover, player 1.
The equilibrium payoff of every other player $i$ is fixed at the
average utility level
\[
\Avg_i:=\int_K U_i(\eta_i)\,d\mu_K(\eta),
\qquad i\in N,
\]
whereas the first mover captures the entire surplus above the average
utility benchmarks. We define fairness among players as equal division
of the surplus above these benchmarks and address the first-mover
advantage by applying the first-mover auction result in
\citet[Proposition~3]{echenique2025price} to our environment.
Theorems~\ref{thm:many-player} and~\ref{thm:max-min} characterize the
continuation payoffs of the first mover and the other players in the
same way as in EN's P\&C, even with an infinite menu and heterogeneous
probability assessments. Hence, EN's bidding argument applies without
modification. At the auction's Nash equilibrium exhibited below,
each player $i$ receives the same surplus increment $\eta/n$ above
their own average utility level $\Avg_i$.

The auction proceeds as follows. Each player $i$ submits a nonnegative
monetary bid $b_i\geq0$. Let
\[
I^*:=\left\{i\in N:\ b_i=\max\{b_1,\ldots,b_n\}\right\}
\]
be the set of highest bidders. One player is drawn uniformly from
$I^*$ and declared the winner. The winner $i$ pays $b_i$, and each
other player receives $b_i/(n-1)$. These transfers take place before
the P$^{n-1}$\&C subgame begins. P\&C then starts with the winner as
the first mover.

By cash invariance, the payments and rebates associated with bidding
enter utility additively. Thus, the auction stage redistributes the
same surplus among players without affecting efficiency.
Choose an efficient allocation
\[
\xi^*\in\arg\max_{\xi\in K}\sum_{i=1}^n U_i(\xi_i)
\]
and define the ex-ante surplus generated by this allocation as
\[
\eta:=\sum_{i=1}^n U_i(\xi_i^*)-\sum_{i=1}^n\Avg_i\geq0.
\]
The bidding argument in EN's Proposition~3 yields the following
corollary.

\begin{corollary}[Based on EN's Proposition~3]\label{cor:first-mover-auction}
Suppose that the assumptions of Theorem~\ref{thm:max-min} hold and that
the normalization measure $\mu_K$ is also fixed as part of the
mechanism and is common knowledge. Players commit to the two-stage
procedure consisting of the auction followed by P$^{n-1}$\&C.
Then the following two statements hold.
\begin{enumerate}
\item In every SPNE, the terminal allocation belongs to
\[
\arg\max_{\xi\in K}\sum_{i=1}^n U_i(\xi_i),
\]
and each player's expected payoff is
\[
\Pi_i^{\mathrm{auction}}
=\Avg_i+\frac{\eta}{n},
\qquad i\in N.
\]
\item Every allocation $\xi^*$ in this argmax set is the terminal
outcome of some SPNE.
\end{enumerate}
Thus, the procedure fully implements the set of optimal risk-sharing
allocations in SPNE and divides the ex-ante surplus above the average
utility benchmarks equally in every SPNE.
\end{corollary}
\noindent
\emph{Verification of applicability.}
By Theorems~\ref{thm:many-player} and~\ref{thm:max-min}, the equilibrium
payoffs in a P\&C subgame with any player $w$ as the first mover are
$\Avg_w+\eta$ for $w$ and $\Avg_j$ for every other player $j$.
Changing the order of play leaves the menu $K$, the normalization
measure $\mu_K$, and each player's utility $U_i$ unchanged, so
$\eta$ is unchanged as well. The zero-sum monetary transfers associated
with bidding only add constants to payoffs in each P\&C subgame and
therefore leave its equilibrium allocations and ex-ante surplus
unchanged. Thus, if $w$ wins the auction, the payoffs after the auction
transfers are
\[
\begin{aligned}
g_w&=\Avg_w+\eta-b_w,\\
g_j&=\Avg_j+\frac{b_w}{n-1}
\qquad(j\ne w).
\end{aligned}
\]
Subtracting each player's average utility level yields the same bidding
game as in the proof of EN's Proposition~3. Hence, EN's bidding-stage
equilibrium argument applies without modification.

Equal division of surplus in every equilibrium can also be verified
as follows. Set
\[
b^*:=\frac{n-1}{n}\eta.
\]
By bidding $b^*$, a player obtains a surplus of $\eta-b^*=\eta/n$
if they win. If they lose, the winning bid satisfies $b_w\geq b^*$,
so their surplus is $b_w/(n-1)\geq\eta/n$.
Thus, each player can secure a surplus of at least $\eta/n$ above
their own average utility level.
At any Nash equilibrium of the bidding stage, each player's expected
surplus is at least this level. Since total surplus always equals
$\eta$, every player's expected surplus must equal $\eta/n$.
Moreover, in every SPNE, the P\&C continuation selects an optimal
allocation by Theorem~\ref{thm:max-min}. This establishes statement~1
of the corollary.

If all players bid $b^*$, each player's surplus equals $\eta/n$
regardless of whether they win or lose. A unilateral increase in the
bid raises the payment upon winning and reduces the deviator's
surplus, whereas a decrease leaves the rebate from another player's
winning bid unchanged. Hence, this bid profile is a Nash equilibrium.
Furthermore, for any efficient allocation $\xi^*$,
Theorems~\ref{thm:many-player} and~\ref{thm:max-min} allow us to specify
a P\&C continuation SPNE that selects $\xi^*$ after every bid history
and for every auction winner. Combining these continuation strategies
with the Nash equilibrium in which all players bid $b^*$ yields an
SPNE of the two-stage game with terminal allocation $\xi^*$,
establishing statement~2 of the corollary.

\section{Conclusion}\label{sec:conclusion}

This paper has analyzed a procedure for implementing optimal risk sharing
through decentralized participants' own choices. We extended EN's P\&C
mechanism, originally proposed for a finite choice set, to a
weak$^\ast$-compact infinite menu allowing allocations to depend on
background factors affecting the aggregate loss.
We showed that, if monetary utilities are weak$^\ast$-continuous on the
menu, the extended P\&C mechanism fully implements, in subgame-perfect
equilibrium, the set of allocations that maximize the sum of
participants' monetary utilities over the menu.

We also showed that this full-implementation result continues to hold
in an environment in which participants have heterogeneous credal sets
and probability measures over those sets. The monetary utilities
formulated in this environment extend the monetary certainty-equivalent
representation of the smooth-ambiguity model of \citet{Klibanoff2005}
to a broader class of utility evaluations,
while retaining exponential ambiguity evaluation.

Applying EN's first-mover auction in our environment also resolves the
inequity inherent in P\&C. We showed that, even when participants
disagree about probability assessments, the extended P\&C procedure
combined with bidding to select the first mover fully implements
optimal and fair risk sharing in subgame-perfect equilibrium.

\begin{appendices}
\counterwithin{theorem}{section}
\counterwithin{lemma}{section}
\counterwithin{proposition}{section}
\numberwithin{equation}{section}

\section{Compactness/metrizability of $K$}\label{sec:appA}
\begin{proof}[Proof of Lemma \ref{lem:Delta-compact}] ~

Let $E:=(L^\infty)^n$ be endowed with the weak$^\ast$ topology
$\sigma(E,(L^1)^n)$.

\medskip
\noindent\emph{Step 1 (boundedness).}
Fix $\xi\in \Delta_{X}$. By definition, $\sum_{i=1}^n \xi_i=X$ a.s.
and each $\xi_i$ shares the sign of $X$.
In particular, on $\{X>0\}$ we have $\xi_i\ge 0$ and hence
$0\le \xi_i \le \sum_{j=1}^n \xi_j = X$.
Similarly, on $\{X<0\}$ we have $\xi_i\le 0$ and hence
$X=\sum_{j=1}^n \xi_j \le \xi_i \le 0$.
On $\{X=0\}$ the sign restriction implies $\xi_i=0$.
Therefore $|\xi_i|\le |X|$ a.s., so $\|\xi_i\|_\infty \le \|X\|_\infty$ for all $i$.
With the product norm $\|\xi\|:=\max_i \|\xi_i\|_\infty$, we obtain
$\Delta_{X}\subset r\,B_E$ where $r:=\|X\|_\infty$.
By Banach--Alaoglu--Bourbaki, $r\,B_E$ is weak$^\ast$-compact.

\medskip
\noindent\emph{Step 2 (closedness).}
Define the linear map
\[
F:E\to L^\infty,\qquad F(\xi):=\sum_{i=1}^n \xi_i.
\]
To see that $F$ is weak$^\ast$-continuous, fix $h\in L^1$ and note that
\[
\langle F(\xi),h\rangle = \sum_{i=1}^n \langle \xi_i,h\rangle
= \big\langle \xi,\,(h,\dots,h)\big\rangle,
\]
and $\xi\mapsto \langle \xi,(h,\dots,h)\rangle$ is continuous by definition of the weak$^\ast$ topology
$\sigma(E,(L^1)^n)$, i.e. the coarsest topology on $E$ making $\xi\mapsto \langle \xi,g\rangle$ continuous for all $g\in (L^1)^n$.
Hence $F$ is weak$^\ast$-continuous and
\[
C:=\{\xi\in E:\ F(\xi)=X\} = F^{-1}(\{X\})
\]
is weak$^\ast$-closed.

Let $A_+:=\{X>0\}$, $A_-:=\{X<0\}$ and $A_0:=\{X=0\}$.
For each $i$ define
\[
B_i^+:=\{\xi\in E:\ 1_{A_+}\xi_i\ge 0\ \mathbb{P}\text{-a.s.}\},\qquad
B_i^-:=\{\xi\in E:\ 1_{A_-}\xi_i\le 0\ \mathbb{P}\text{-a.s.}\}.
\]
Using the characterization $Y\ge 0$ $\mathbb{P}$-a.s. $\Leftrightarrow$ $\langle Y,h\rangle:=\int Yh d\mathbb{P} \geq 0$
for all $h \in L^1_+(\mathbb{P})$, we can write
\[
B_i^+ = \bigcap_{h\in L^1_+}\{\xi:\ \langle \xi,(0,\dots,1_{A_+}h,\dots,0)\rangle\ge 0\},
\]
\[
B_i^- = \bigcap_{h\in L^1_+}\{\xi:\ \langle \xi,(0,\dots,-1_{A_-}h,\dots,0)\rangle\ge 0\},
\]
so $B_i^\pm$ are weak$^\ast$-closed as intersections of weak$^\ast$-closed half-spaces.
Note that, for any measurable set $A$, multiplication by $1_A$ maps
$L^1$ into $L^1$. Thus, $1_{A_+}h \in L^1_+$ for $h \in L^1_{+}$.
In addition, 
\[
B_i^0:=\{\xi:\ 1_{A_0}\xi_i=0\}
      =\bigcap_{h\in L^1}\{\xi:\ \langle \xi,(0,\dots,1_{A_0}h,\dots,0)\rangle=0\},
\]
which is weak$^\ast$-closed.
Set $B_i:=B_i^+\cap B_i^-\cap B_i^0$ and $B:=\cap_{i=1}^n B_i$.
Then $B$ is weak$^\ast$-closed and
\[
\Delta_{X}= C\cap B
\]
is weak$^\ast$-closed in $E$.

Since $\Delta_{X}$ is weak$^\ast$-closed and contained in the weak$^\ast$-compact
ball $r\,B_E$, it is weak$^\ast$-compact.
\end{proof}
\section{Pareto optimality and sub-convolution}\label{sec:appB}
\begin{proof}[Proof of Theorem \ref{thm:sup-convolution-pareto-optimal}]~ %

\noindent
(ii) $\Rightarrow$ (i): Suppose, contrary to our claim, that 
\[\exists (\zeta_i) \text{ with } U_i(\zeta_i) \geq U_i(\xi_i) \text{ such that } \exists i \text{ with }U_i(\zeta_i) > U_i(\xi_i). \]
Then we have 
\[\sum_{i \in N} U_i(\zeta_i) > \sum_{i \in N} U_i(\xi_i),\]
a contradiction.

\vspace{12pt}
 
 (i) $\Rightarrow$ (ii): 
 Let $\tilde{B}$ denote the set 
 \[\tilde{B}:= \{(U_i(\zeta_i))_{i \in N} \ \vert \ (\zeta_i)_{i \in N} \in \mathbb{A}(X)\},\]
and by $B$ the set 
\[B:=\tilde{B} - \mathbb{R}^n_{+},\]
where the "$-$" between the sets is the Minkowski difference. Furthermore, let $C$ be the set 
\[C:=\{(U_i(\xi_i))_{i \in N}\} + (0,\infty)^n,\]
where the "$+$" between the sets is the Minkowski sum. 
Clearly, $C$ is an open convex set and $B, C \ne \emptyset$.
We claim that $B$ is convex. Take arbitrary $x,x'\in B$.
By the definition of $B$, there exist
$\zeta,\zeta'\in\mathbb A(X)$ and
$r,s\in\mathbb R_+^n$ such that
\[
x=(U_i(\zeta_i))_{i\in N}-r,
\qquad
x'=(U_i(\zeta_i'))_{i\in N}-s.
\]
For $\alpha\in[0,1]$, set
\[
\bar\zeta
=
\alpha\zeta+(1-\alpha)\zeta'
\in\mathbb A(X).
\]
By the concavity of each $U_i$,
\[
q
:=
(U_i(\bar\zeta_i))_{i\in N}
-
\alpha(U_i(\zeta_i))_{i\in N}
-
(1-\alpha)(U_i(\zeta_i'))_{i\in N}
\in\mathbb R_+^n.
\]
Hence,
\[
\begin{aligned}
\alpha x+(1-\alpha)x'
&=
(U_i(\bar\zeta_i))_{i\in N}
-
\{q+\alpha r+(1-\alpha)s\}\\
&\in
\widetilde B-\mathbb R_+^n
=
B.
\end{aligned}
\]
Thus, $B$ is convex.

We claim that $B \cap C = \emptyset$. It follows from the Pareto optimality of $(\xi_i)_{i \in N}$. By the Hahn-Banach theorem (first-geometric form), there exists a closed hyperplane $H=[f=\alpha]$ which separates $B$ and $C$, where $f$ is a linear function on $\mathbb{R}^n$ and $\alpha \in \mathbb{R}$:
\begin{align}
\sup_{b \in B} \lambda \cdot b \leq \alpha \leq \inf_{c \in C} \lambda \cdot c. \label{eq:hyperplane}
\end{align}
Note that $u^{\ast}:=(U_i(\xi_i))_{i\in N} \in B$ and $(U_i(\xi_i)+y_i)_{i \in N} \in C$ for all $y_i  > 0$, $i \in N$. 
We deduce $\lambda_i \geq 0$ from the separation inequality. 
From \eqref{eq:hyperplane} for $t>0$
\[
\sup_{b \in B} \lambda  \cdot b \leq \alpha \leq \inf_{c \in C} \lambda \cdot c \leq \lambda \cdot u^{\ast} + t \lambda \cdot y.
\] 
Letting $t \rightarrow 0$ in the preceding inequality, we obtain 
\[\sup_{b \in B} \lambda \cdot b = \alpha = \lambda \cdot u^{\ast}.\]
Thus, the hyperplane $H$ is tangent to $B$ at $u^{\ast}$.
Let $D:=\{c \in \mathbb{R}^n \vert \sum_i c_i =0\}$. 
By the cash invariance of $U_i$ and the construction of $\mathbb{A}(X)$, we have  
\[\tilde{B}=\tilde{B}+D \quad \text{ and } \quad B=B+D.\]
We claim $\lambda \in D^{\bot}$. For any $c \in D$ and for any $b\in B$ 
we must have 
\[
\lambda \cdot (b+tc) \leq \alpha \quad \forall t \in \mathbb{R}.
\] This forces $\lambda \cdot c=0$ for all $c \in D$. 
Hence, $\lambda \in D^{\bot}$, i.e., $\lambda$ is proportional to $(1,\cdots, 1)$. We rescale $\lambda=c_0(1,\cdots, 1)$ with $c_0>0$ to $\lambda=(1,\cdots, 1)$. 
Therefore, (ii) holds, and we see 
\[
\mathrm{1} \cdot u^{\ast}=\sum U_i(\xi_i)=U_1\square \cdots \square U_n(X).
\]
\end{proof}
\section{Best-response compactness lemma}\label{sec:appD}
\begin{proof}[Proof of Lemma \ref{lem:closed-set-A_2}] ~

From Lemma \ref{lem:Delta-compact}, $K$ is weak$^{\ast}$-compact.
Because $p \in P_2$ and $U_2$ are weak$^{\ast}$-continuous, the map $f=U_2-p$ is weak$^{\ast}$-continuous on $K$. 
On a compact space, a continuous function attains a maximum. Hence, $\mathbb{A}_p^2\neq\emptyset$. 
Write $M:=\sup_{\xi\in K} f(\xi)$. Then $\mathbb{A}_p^2$ is rewritten as 
\[
\mathbb{A}_p^2 \;=\; \{\xi\in K:\ f(\xi)=M\}.
\]
Since $\mathbb{A}^2_p$ is the preimage of the singleton set of 
$\{M\}$, it is a closed subset of $K$. 
Therefore, $\mathbb{A}_p^2$ is a closed subset of the compact set $K$, and thereby compact. 
\end{proof}
\section{The ba-space/weak$^{\ast}$ compactness lemma}\label{sec:appE}
\begin{proof}[Proof of Lemma \ref{lem:credal-set-compactness}] ~

For any $\mu \in ba(\Omega, \mathcal{F})$, the norm satisfies $\|\mu\|_{ba}=|\mu|(\Omega)$ (total variation). Thus, if $\mu$ is a probability charge, $\|\mu\|_{ba}=1$. Hence
\[\mathcal D_{\mathbb P} \subset B(ba):=\{\mu \in ba(\Omega, \mathcal{F}) \vert \|\mu\|_{ba} \leq 1\}.\]
By the Banach–Alaoglu–Bourbaki theorem, $B(ba)$ is compact in the weak$^{\ast}$ topology $\sigma(ba, B_b)$. 

\vspace{12pt}
\noindent
The constraints $\mu(\Omega)=1$ and $\mu\ge 0$ are weak$^{\ast}$-closed: for each $A\in \mathcal{F}$,
$\mu\mapsto \mu(A)=\langle \mathbf 1_A,\mu\rangle$ is continuous, hence
$\{\mu:\mu(\Omega)=1\}$ and $\bigcap_{A}\{\mu:\mu(A)\ge 0\}$ are closed. 
Domination $\mu\ll\mathbb P$ is $\bigcap_{N \in \mathcal{F}:\mathbb P(N)=0}\{\mu:\mu(N)=0\}$, also closed.

Consequently, $\mathcal D_{\mathbb P}$ is a weak$^{\ast}$-closed subset
of the weak$^{\ast}$-compact ball $B(ba)$, and hence is weak$^{\ast}$-compact.
By definition, each $\mathcal P_i$ is a weak$^{\ast}$-closed subset of
$\mathcal D_{\mathbb P}$, so it is weak$^{\ast}$-compact as well.
\end{proof}

\section{Properties of exponential smooth-ambiguity monetary utilities}\label{sec:appF}
\begin{proof}[Proof of Lemma \ref{lem:max-min-properties}] ~
Fix $i\in N$ and write
\[
I_i(Y):=
\int_{\mathcal P_i}
\exp\left(-\frac{V_{i,\nu}(Y)}{\tau_i}\right)
\,\Lambda_i(d\nu),
\qquad U_i(Y)=-\tau_i\log I_i(Y).
\]
Measurability and the uniform bound
$|V_{i,\nu}(Y)|\leq\|Y\|_\infty$, together with
$\Lambda_i(\mathcal P_i)=1$, give
\[
0<e^{-\|Y\|_\infty/\tau_i}
\leq I_i(Y)
\leq e^{\|Y\|_\infty/\tau_i}<\infty.
\]
Thus $U_i(Y)$ is real valued and satisfies
$|U_i(Y)|\leq\|Y\|_\infty$.

Normalization follows from $I_i(0)=1$.
If $Y\geq Z$ $\mathbb P$-a.s., monotonicity of each $V_{i,\nu}$
gives $I_i(Y)\leq I_i(Z)$, and hence $U_i(Y)\geq U_i(Z)$.
For every $c\in\mathbb R$, cash invariance of each $V_{i,\nu}$ gives
\[
I_i(Y+c)=e^{-c/\tau_i}I_i(Y),
\]
so $U_i(Y+c)=U_i(Y)+c$.

For concavity, fix $Y,Z\in L^\infty$ and $t\in(0,1)$.
Concavity of the prior-specific utilities, followed by H\"older's
inequality with exponents $1/t$ and $1/(1-t)$, yields
\[
\begin{aligned}
I_i(tY+(1-t)Z)
&\leq
\int_{\mathcal P_i}
\left(e^{-V_{i,\nu}(Y)/\tau_i}\right)^t
\left(e^{-V_{i,\nu}(Z)/\tau_i}\right)^{1-t}
\,\Lambda_i(d\nu)\\
&\leq I_i(Y)^t I_i(Z)^{1-t}.
\end{aligned}
\]
Taking $-\tau_i\log$ gives
\[
U_i(tY+(1-t)Z)\geq tU_i(Y)+(1-t)U_i(Z).
\]
The cases $t=0$ and $t=1$ are immediate.

To prove allocation continuity, fix $\xi^0\in K$ and $\varepsilon>0$.
The map
\[
G_i(\xi,\nu):=F_i(\xi,\nu)-F_i(\xi^0,\nu)
\]
is jointly continuous and satisfies $G_i(\xi^0,\nu)=0$ for every
$\nu\in\mathcal P_i$. For each $\nu^0\in\mathcal P_i$, choose open
neighborhoods $N_{\nu^0}$ of $\xi^0$ in $K$ and $O_{\nu^0}$ of
$\nu^0$ in $\mathcal P_i$ such that
\[
|G_i(\xi,\nu)|<\frac{\varepsilon}{2}
\qquad\text{on }N_{\nu^0}\times O_{\nu^0}.
\]
By Lemma~\ref{lem:credal-set-compactness}, $\mathcal P_i$ is compact,
so there is a finite subcover
$\mathcal P_i=O_{\nu^1}\cup\cdots\cup O_{\nu^m}$.
The finite intersection $N:=\bigcap_{r=1}^mN_{\nu^r}$ is an open
neighborhood of $\xi^0$. For every $\xi\in N$ and every
$\nu\in\mathcal P_i$, some $r$ satisfies $\nu\in O_{\nu^r}$, while
$\xi\in N_{\nu^r}$. Hence
\[
|F_i(\xi,\nu)-F_i(\xi^0,\nu)|<\frac{\varepsilon}{2}
\qquad\text{for all }\xi\in N,\ \nu\in\mathcal P_i.
\]
Integrating the corresponding bounds on the exponentials yields
\[
e^{-\varepsilon/(2\tau_i)}I_i(\xi_i^0)
\leq I_i(\xi_i)
\leq e^{\varepsilon/(2\tau_i)}I_i(\xi_i^0).
\]
Taking $-\tau_i\log$ therefore gives
\[
U_i(\xi_i^0)-\frac{\varepsilon}{2}
\leq U_i(\xi_i)
\leq U_i(\xi_i^0)+\frac{\varepsilon}{2}.
\]
Thus $|U_i(\xi_i)-U_i(\xi_i^0)|<\varepsilon$ for every $\xi\in N$,
which proves weak$^\ast$ continuity on $K$.
\end{proof}
\section{A norm-compact example for the P\&C mechanism}\label{app:entropic-simplex}
In this example only, replace the full same-sign menu by
\[
\begin{gathered}
S_{n-1}:=\left\{\alpha\in[0,1]^n:
\sum_{i=1}^n\alpha_i=1\right\},\\
K_{\mathrm{prop}}
:=
\left\{(\alpha_1X,\ldots,\alpha_nX):
\alpha\in S_{n-1}\right\}
\subset\Delta_X.
\end{gathered}
\]
The coefficients $\alpha_i$ are deterministic, while the allocations
$\alpha_iX$ remain random. Thus, this example restricts the admissible
contracts to proportional sharing; the general model continues to allow
arbitrary measurable random sharing weights.
Equip $(L^\infty)^n$ with the product norm
$\|\xi\|_{\infty,n}:=\max_i\|\xi_i\|_\infty$. 
For
$F_X(\alpha):=(\alpha_1X,\ldots,\alpha_nX)$,
\[
\|F_X(\alpha)-F_X(\beta)\|_{\infty,n}
=\|X\|_\infty\max_i|\alpha_i-\beta_i|.
\]
Since $F_X$ is continuous and $S_{n-1}$ is compact, its image
$F_X(S_{n-1})=K_{\mathrm{prop}}$ is norm compact.
Moreover, $\|X\|_\infty>0$ makes $F_X$ injective, so
$K_{\mathrm{prop}}$ is an $(n-1)$-dimensional simplex.

Let all agents evaluate payoffs under the common probability measure
$\mathbb P$, with entropic monetary utilities
\[
U_i(Y)
=-\frac{1}{\gamma_i}\log\mathbb E
\left[e^{-\gamma_iY}\right],
\qquad \gamma_i>0.
\]
These utilities are concave, monotone, cash-invariant, and normalized by
$U_i(0)=0$.

Each $U_i$ is $1$-Lipschitz, hence norm-continuous,
on $L^\infty(\mathbb P)$. In particular,
\[
|U_i(\alpha_iX)-U_i(\beta_iX)|
\leq\|X\|_\infty|\alpha_i-\beta_i|.
\]
Thus, the maps $\xi\mapsto U_i(\xi_i)$ are norm-continuous on
$K_{\mathrm{prop}}$. Since this menu is norm compact, the norm topology
and the relative weak$^\ast$ topology inherited from $(L^\infty)^n$
coincide on $K_{\mathrm{prop}}$.
These restricted maps are therefore also continuous for the relative
weak$^\ast$ topology on $K_{\mathrm{prop}}$, as required by
Assumption~\ref{ass:weakstar-continuity} for this example's menu.

For the P\&C mechanism described in Section~\ref{sec:P&C}, use $K=K_{\mathrm{prop}}$ and
take $\mu_K=\mu_{\mathrm{prop}}$. 
 Define $\mu_{\mathrm{prop}}$ as the pushforward under $F_X$
of the uniform probability measure on $S_{n-1}$. 
This measure has
full support on $K_{\mathrm{prop}}$. All price normalizations and
utility averages in the example are taken over this menu using
$\mu_{\mathrm{prop}}$.

An explicit welfare maximizer is obtained by writing
\[
\tau_i:=\gamma_i^{-1},
\qquad T:=\sum_{j=1}^n\tau_j,
\qquad
\alpha_i^\ast:=\frac{\tau_i}{T},
\qquad \xi_i^\ast:=\alpha_i^\ast X.
\]
Indeed, for any attainable allocation $\sum_i\xi_i=X$, H\"older's
inequality and then taking logarithms give
\[
\begin{aligned}
\mathbb E[e^{-X/T}]
&\leq
\prod_{i=1}^n
\left(\mathbb E[e^{-\gamma_i\xi_i}]\right)^{\tau_i/T},\\
\sum_{i=1}^n U_i(\xi_i)
&\leq-T\log\mathbb E[e^{-X/T}].
\end{aligned}
\]
The proportional allocation $\xi^\ast$ attains equality and belongs to
$K_{\mathrm{prop}}$. Thus, although implementation on a restricted menu
generally concerns efficiency relative to that menu and its deterministic
zero-sum transfers, the restriction loses no aggregate welfare in this
common-prior entropic example.
\end{appendices}

\bibliography{jecon_ref3}
\end{document}